%% file: lata2016.tex
\documentclass{article}
\usepackage{tikz}
\usepackage{authblk}
\usepackage{array}
\newcolumntype{P}[1]{>{\centering\arraybackslash}p{#1}}

\usepackage{etoolbox}
\usepackage{amssymb,amsmath,amsthm}
\usepackage[hidelinks,
colorlinks=true,
linkcolor=myRed,
urlcolor=blue,
citecolor=myBlue
]{hyperref}
\usepackage{bookmark}
\usepackage{aliascnt}
\usepackage[margin=3.5cm,footskip=2cm]{geometry}

\newtheorem{thm}{Theorem}{\bfseries}{\itshape}

\newtheorem{definition}[thm]{Definition}
\newaliascnt{lem}{thm}
\newtheorem{lem}[lem]{Lemma}{\bfseries}{\itshape}
\aliascntresetthe{lem}	

\DeclareMathAlphabet\EuRoman{U}{eur}{m}{n}
\SetMathAlphabet\EuRoman{bold}{U}{eur}{b}{n}

\input{macros}

\begin{document}
\title{Hankel Matrices for Weighted \\ Visibly Pushdown Automata}
\author[1]{\normalsize Nadia Labai\thanks{Supported by the National Research Network RiSE (S114), and the LogiCS doctoral program (W1255) funded by the Austrian Science Fund (FWF).}}
\author[2]{\normalsize Johann A.~Makowsky\thanks{Partially supported by a grant of Technion Research Authority.}}
\affil[1]{Department of Informatics, Vienna University of Technology \texttt{\small labai@forsyte.at}}
\affil[2]{Department of Computer Science, Technion - Israel Institute of Technology {\texttt{\small janos@cs.technion.ac.il}}}	

\renewcommand\Authands{ and }
\date{}

\maketitle
\setcounter{footnote}{0}

\begin{abstract}
Hankel matrices (aka connection matrices) of word functions and graph
parameters have wide applications in automata theory, graph theory, and machine
learning. We give a characterization of real-valued functions on nested words
recognized by weighted visibly pushdown automata in terms of Hankel matrices
on nested words. This complements C. Mathissen's characterization in terms
of weighted monadic second order logic.
\end{abstract}

\input{intro}
\input{prelim}
\input{learning}
\input{conc}

\input{nested}

\subsubsection*{Acknowledgments.} We thank Boaz Blankrot for helpful discussions on matrix decompositions and the anonymous referees for valuable feedback.

\bibliographystyle{abbrv}
\bibliography{lata2016_ref}

\end{document}

%% file: macros.tex
\definecolor{myRed}{RGB}{200,0,0}

\definecolor{myBlue}{RGB}{0,0,200}

\newcommand{\bfalpha}{\boldsymbol\alpha}
\newcommand{\bfeta}{\boldsymbol\eta}

\newcommand{\lb}{\left}
\newcommand{\rb}{\right}
\newcommand{\lng}{\langle}
\newcommand{\rng}{\rangle}

\newcommand{\cB}{\mathcal{B}}

\newcommand{\cF}{\mathcal{F}}

\newcommand{\cS}{\mathcal{S}}

\newcommand{\bR}{\mathbb{R}}
\newcommand{\bC}{\mathbb{C}}

\newcommand{\bN}{\mathbb{N}}

\newcommand{\bfv}{\mathbf{v}}

\newcommand{\bfx}{\mathbf{x}}

\newcommand{\bfy}{\mathbf{y}}

\newcommand{\MSOLEVAL}{\mathrm{MSOLEVAL}}
\newcommand{\WMSOL}{\mathrm{WMSOL}}

\newcommand{\Ha}{\mathbf{H}}
\newcommand{\nHa}{\mathbf{nH}}

\newcommand{\bfA}{\mathbf{A}}
\newcommand{\bfM}{\mathbf{M}}
\newcommand{\bfN}{\mathbf{N}}
\newcommand{\bfI}{\mathbf{I}}
\newcommand{\bfr}{\mathbf{r}}
\newcommand{\bfp}{\mathbf{p}}

\newcommand{\WNW}{\mathrm{WNW}}
\newcommand{\NW}{\mathrm{NW}}

%% file: intro.tex
\section{Introduction and Background}
\subsection{Weighted Automata for Words and Nested Words}
Classical word automata can be extended to \emph{weighted} word automata by assigning weights from some numeric domain to their transitions, thereby having them assign values to their input words
rather than accepting or rejecting them.
Weighted (word) automata define the class of recognizable word functions, first introduced in the study of stochastic automata
by A. Heller \cite{ar:Heller1967}.
Weighted automata are used
in verification,
\cite{bk:Arnold1994,bk:McMillan1993},
in program synthesis,
\cite{pr:ChatterjeeDH09,pr:ChatterjeeHJS10},
in digital image compression,
\cite{ar:CulikKari93}, and 
speech processing,
\cite{ar:Mohri97,ar:Fernando97,ar:Allauzen04}. For a comprehensive survey, see the Handbook of Weighted Automata \cite{bk:DrosteKuichVogler2009}.
Recognizable word functions over commutative semirings $\cS$ 
were characterized using logic through the formalism of Weighted Monadic Second Order Logic ($\WMSOL$), \cite{ar:DrosteGastin05}, 
and the formalism of $\MSOLEVAL$
\footnote{This formalism was originally introduced in \cite{ar:CourcelleMakowskyRoticsDAM} for graph parameters.}, \cite{pr:LabaiMakowsky2013}.

Nested words and nested word automata are generalizations of words and finite automata, introduced by Alur and Madhusudan \cite{ar:AlurMadhusudan06}.
A \emph{nested word}  $nw \in \NW(\Sigma)$ over an alphabet $\Sigma$ is a sequence of linearly ordered positions, augmented with forward-oriented edges that do not cross, creating a nested structure.
In the context of formal verification for software, execution paths in procedural programs are naturally modeled by nested words
whose hierarchical structure captures calls and returns. Nested words also model annotated linguistic data and 
tree-structured data which is given by a linear encoding, such as HTML/XML documents. Nested word automata define the class of regular languages of nested words. The key feature of these automata is their ability to propagate hierarchical states along the augmenting edges, in addition to the states propagated along the edges of the linear order. We refer the reader to \cite{ar:AlurMadhusudan06} for details.
Nested words $nw \in \NW(\Sigma)$ can be (linearly) encoded as words over an extended tagged alphabet $\hat{\Sigma}$, where
the letters in  $\hat{\Sigma}$
specify whether the position is a call, a return, or neither (internal). 
Such encodings of regular languages of nested words give the class of \emph{visibly pushdown languages} over the tagged alphabet $\hat{\Sigma}$,
which lies between the parenthesis languages and deterministic context-free languages. 
The accepting pushdown automata for visibly pushdown languages push one symbol when reading 
a call, pop one symbol when reading a return, and only update their control when reading an internal symbol. Such automata are called \emph{visibly pushdown automata}.
Since their introduction, nested words and their automata 
have found applications in specifications for program analysis \cite{driscoll2011checking,harris2012secure,driscoll2012opennwa}, XML processing \cite{gauwin2011streamable,mozafari2012high}, and have motivated several theoretical questions, \cite{d2014symbolic,alur2007first,murawski2005third}.

Visibly pushdown automata and nested word automata were extended by assigning weights from a commutative semiring $\cS$ to their transitions as well.
Kiefer et al introduced \emph{weighted visibly pushdown automata}, and their equivalence problem was showed to be logspace reducible to polynomial identity testing, \cite{ar:Kiefer12}.
Mathissen introduced \emph{weighted nested word automata}, and proved 
a logical characterization of their functions using a modification of $\WMSOL$, \cite{ar:Mathissen08}.

\subsection{Hankel Matrices and Weighted Word Automata}
Given a word function $f: \Sigma^\star \rightarrow \cF$, its \emph{Hankel matrix} $\Ha_f \in \cF^{\Sigma^\star \times \Sigma^\star}$
is the infinite matrix whose rows and columns are indexed by words in $\Sigma^\star$ and $\Ha_f(u,v) = f(uv)$, where $uv$ is the concatenation of $u$ and $v$.
In addition to the logical characterizations, 
there exists a characterization of recognizable word functions via Hankel matrices, by Carlyle and Paz \cite{ar:CarlylePaz1971}.
\begin{thm}[Carlyle and Paz, 1971]
\label{th:carlyle_paz}
A real-valued word function $f$ is recognized by a weighted (word) automaton iff $\Ha_f$ has finite rank.
\end{thm}
The theorem was originally stated using the notion
of external function rank, but the above formulation is equivalent.
Multiplicative words functions were characterized by Cobham \cite{ar:Cobham1978} as exactly those with a Hankel matrix of rank $1$.

Hankel matrices proved useful also in the study of graph parameters. Lov\'asz introduced a kind of Hankel matrices for graph parameters \cite{ar:Lovasz07}
which were used
to study real-valued graph parameters and their relation to partition functions,
\cite{ar:FreedmanLovaszSchrijver07,bk:Lovasz-hom}.
In \cite{ar:GodlinKotekMakowsky08},
the definability of graph parameters in monadic second order logic was related to the rank of their Hankel matrices.
Meta-theorems involving logic, such as Courcelle's theorem and generalizations thereof \cite{bk:DowneyFellows99,bk:CourcelleEngelfriet2011,pr:CourcelleMakowskyRoticsWG98,ar:MakowskyTARSKI}, were made logic-free by replacing their definability conditions with conditions on Hankel matrices, \cite{ar:LabaiMakowsky2014,msc:Labai,ar:LabaiMakowskyJCSS}.

\subsection{Our Contribution}
The goal of this paper is to prove a characterization of the functions recognizable by weighted visibly pushdown automata (WVPA), 
called here \emph{recognizable nested word functions}, via Hankel matrices.
Such a characterization would nicely fill the role of the Carlyle-Paz theorem in the words setting, complementing
results that draw parallels between recognizable word functions and nested word functions, 
such as the attractive properties of closure and decidability the settings share \cite{ar:AlurMadhusudan06},
and the similarity between the $\WMSOL$-type formalisms used to give their logical characterizations.

The first challenge is in the choice of the Hankel matrices at hand.
A naive straightforward adaptation of the Carlyle-Paz theorem to the setting of nested words would involve Hankel matrices 
for words over the extended alphabet $\hat{\Sigma}$ with the usual concatenation operation on words.
However, then we would have functions recognizable by WVPA with Hankel matrices of infinite rank.
Consider the Hankel matrix of the characteristic function 
of the language of balanced brackets, also known as the Dyck language. 
This language is not regular, so its characteristic function is not recognized by a weighted word automaton.
Hence, by the \hyperref[th:carlyle_paz]{Carlyle-Paz theorem}, its Hankel matrix would have infinite rank despite the fact its 
encoding over a tagged alphabet 
is recognizable by VPA, hence also by WVPA.
\subsubsection*{Main results}
We introduce \emph{nested Hankel matrices} over \emph{well-nested words} (see \autoref{sec:prelim}) to overcome the 
point described above
and prove 
the following characterization of WVPA-recognizable functions of well-nested words:
\begin{thm}[Main Theorem]
\label{th:nested}\ \\
Let $\cF = \bR$ or $\cF = \bC$, and let $f$ be an $\cF$-valued function on well-nested words. Then $f$ is recognized by a weighted visibly pushdown automaton
with $n$ states iff the nested Hankel matrix $\nHa_f$ has rank $\leq n^2$. 
\end{thm}
As opposed to the characterizations of word functions, which allow $f$ to have values over a semiring, we require that $f$ is over $\bR$ or $\bC$. 
This is due to the second challenge, which stems from the fact that in our setting of functions of well-nested words, 
the helpful decomposition properties exhibited by Hankel matrices for word functions are absent. 
This is because, as opposed to words, well-nested words cannot be 
split in arbitrary positions and result in two well-nested words.
Thus, we use the \hyperref[th:SVD]{singular value decomposition (SVD) theorem}, see, e.g., \cite{bk:Golub12}, which is valid only over $\bR$ and $\bC$.

\subsubsection*{Outline}
In \autoref{sec:prelim} we complete the background on well-nested words and weighted visibly pushdown automata, and introduce nested Hankel matrices. The rather technical proof of \autoref{th:nested} is given in \autoref{sec:nested}. In \autoref{sec:learn} we discuss the applications of \autoref{th:nested} to learning theory.
In \autoref{sec:conc} we briefly discuss limitations of our methods and possible extensions of our characterization. 

%% file: prelim.tex
\section{Preliminaries}
\label{sec:prelim}
For the remainder of the paper, we assume that $\cF$ is $\bR$ or $\bC$.
Let $\Sigma$ be a finite alphabet. For $\ell \in \bN^+$, we denote the set $\{1,\dots,\ell\}$ by $[\ell]$. For a matrix or vector $\bfN$, denote its transpose by $\bfN^T$.
Vectors are assumed to be column vectors unless stated otherwise.
\subsection{Well-Nested Words}
We follow the definitions in \cite{ar:AlurMadhusudan06} and \cite{ar:Mathissen08}.
\label{def:wnw}
A \emph{well-nested word} over $\Sigma$ is a pair $(w,\nu)$ where $w \in \Sigma^\star$ of length $\ell$
and $\nu$ is a matching relation for $w$. A \emph{matching relation}\footnote{The original definition of nested words allowed ``dangling" edges. We will only be concerned with nested words that are well-matched.} 
 for a word of length $\ell$ is a set of edges $\nu \subset [\ell] \times [\ell]$
such that the following holds:
\begin{enumerate}
\setlength\itemsep{0.5em}
\item 
If $(i,j) \in \nu$, then $i < j$.
\item
Any position appears in an edge of $\nu$ at most once: 
For $1 \leq i \leq \ell$,\\ 
$| \{ j \mid (i,j) \in \nu\} | \leq 1$ and $| \{ j \mid (j,i) \in \nu\} | \leq 1$
\item
If $(i,j), (i',j') \in \nu$, then it is not the case that $i < i' \leq j < j'$. That is, the edges do not cross.
\end{enumerate}
Denote the set of well-nested words over $\Sigma$ by $\WNW(\Sigma)$.

Given positions $i,j$ such that $(i,j) \in \nu$, position $i$ is a \emph{call} position and position $j$ is a \emph{return} position.
Denote $\Sigma_{{call}} = \{\lng s  \mid s \in \Sigma\}$,\ $\Sigma_{{ret}} = \{s \rng  \mid s \in \Sigma\}$, 
and $\hat{\Sigma} = \Sigma_{{call}} \cup \Sigma_{{ret}} \cup \Sigma_{{int}}$ where $\Sigma_{{int}} = \Sigma$ and is disjoint from $\Sigma_{{call}}$ and $\Sigma_{{ret}}$.
By viewing calls as opening parentheses and returns as closing parentheses, one can define an encoding 
taking nested words over $\Sigma$ to words
over $\hat{\Sigma}$ by assigning to a position labeled $s \in \Sigma$:
\begin{itemize}
\item 
the letter $\lng s$, if it is a call position,
\item
the letter $s \rng$, if it is a return position,
\item
the same letter $s$, if it is an internal position.
\end{itemize}
We denote this encoding by $nw\_w: \WNW(\Sigma) \rightarrow \hat{\Sigma}^\star$ and give an example in \autoref{fig:nested}. 
Note that any parentheses appearing in such an encoding will be well-matched (balanced) parentheses.
Denote its partial inverse function, defined only for words with well-matched parentheses, by $w\_nw : \hat{\Sigma}^\star \rightarrow \WNW(\Sigma)$.
See \cite{ar:AlurMadhusudan06} for details. We will freely pass between the two forms.

\begin{figure}%
\centering
\begin{tikzpicture}[scale=1]
\tikzset{VertexStyle/.append style = {line width = 0.5pt, minimum size = 14pt}}
\tikzstyle{every node} = [draw,shape=circle,color=white]
	\node (1) at (0,0){${\color{black}b}$};
	\node (2) at (1,0){${\color{black}a}$};
	\node (3) at (2,0){${\color{black}a}$};
	\node (4) at (3,0){${\color{black}b}$};
	\node (5) at (4,0){${\color{black}b}$};
	\path [thick, ->] (1) edge (2);
	\path [thick, ->] (2) edge (3);
	\path [thick, ->] (3) edge (4);
	\path [thick, ->] (4) edge (5);
	\path [dashed, thick, out=50, in=130, ->] (2) edge (5);
	\path [dashed, thick, out=50, in=130, ->] (3) edge (4);
	\begin{scope}
	\tikzstyle{every node} = [shape=circle]
	\node at (7,0){$b \lng a \lng ab \rng b\rng$};	
	\end{scope}
\end{tikzpicture}
\caption{On the left, a well-nested word, where the successor relation of the linear order 
is in bold edges, the matching relation is in dashed edges. On the right, its encoding as a word over a tagged alphabet. }
\label{fig:nested}
\end{figure}
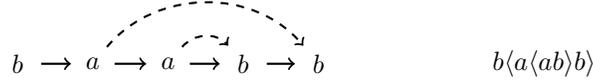

Given a function $f : \WNW(\Sigma) \rightarrow \cF$ on well-nested words,
one can naturally define a corresponding function $f': \hat{\Sigma}^\star \rightarrow \cF$ on words with well-matched parentheses by setting $f'(w) = f(w\_nw(w))$. 
We will denote both functions by $f$. 

\subsection{Nested Hankel Matrices}
Given a function on well-nested words $f : \WNW(\Sigma) \rightarrow \cF$, define its \emph{nested Hankel matrix} $\nHa_f$
as the infinite matrix whose rows and columns are indexed by \emph{words} over $\hat{\Sigma}$ with \emph{well-matched parentheses}, and
$\nHa_f(u,v) = f(uv)$. That is, the entry at the row labeled with $u$ and the column labeled with $v$ is the value $f(uv)$.
A nested Hankel matrix $\nHa_f$ has finite rank if there is a finite set of rows in $\nHa_f$ that linearly span it.
We stress the fact that $\nHa_f$ is defined over words whose parentheses are well-matched, as this is crucial for the proof of \autoref{th:nested}.

\begin{figure}%
\centering
\begin{tabular}{P{1.15cm}P{.9cm}P{.9cm}P{.9cm}P{.9cm}P{.9cm}P{1.15cm}P{.9cm}}
& $\varepsilon$ & $a$ & $\lng a a \rng$ & $aa$ & $\lng a a a \rng$ & $\lng a \lng a a \rng a \rng$ & $\cdots$ \\
$\varepsilon$	& $0$ & $0$ & $1$ & $0$ & $1$ & $2$ & $\cdots$ \\
$a$	& $0$ & $0$ & $1$ & $0$ & $1$ & $2$ & $\cdots$ \\
$\lng a a \rng$	& $1$ & $1$ & $2$ & $1$ & $2$ & $3$ & $\cdots$ \\
$aa$	& $0$ & $0$ & $1$ & $0$ & $1$ & $2$ & $\cdots$\\
$\lng a a a \rng$	& $1$ & $1$ & $2$ & $1$ & $2$ & $3$ & $\cdots$ \\
$\lng a \lng a a \rng a \rng$	& $2$ & $2$ & $3$ & $2$ & $3$ & $4$ & $\cdots$ \\
$\vdots$	& $\vdots$ & $\vdots$ & $\vdots$ & $\vdots$ & $\vdots$ & $\vdots$ & 
\end{tabular}

\caption{The nested Hankel matrix $\nHa_f$. Note that $\nHa_f$ has rank $2$.}%
\label{fig:hankel}%
\end{figure}

As an example, consider the function $f$ which counts the number of pairs of parentheses in a well-nested word over the alphabet $\Sigma = \{a\}$. Then the corresponding 
word function is on words over the tagged alphabet $\hat{\Sigma} = \{a ,\lng a, a \rng \}$.
In \autoref{fig:hankel} we see (part of) the corresponding nested Hankel matrix $\nHa_f$ with labels on its columns and rows.

\subsection{Weighted Visibly Pushdown Automata}
For notational convenience, now let $\Sigma = \Sigma_{{call}} \cup \Sigma_{{ret}} \cup \Sigma_{{int}}$.
We follow the definition given in \cite{ar:Kiefer12}. 
\begin{definition}[Weighted visibly pushdown automata]
An \emph{$\cF$-weighted VPA} on $\Sigma$ is a tuple 
$
A = (n, \bfalpha, \bfeta, \Gamma, \mathrm{M})
$
where 
\begin{itemize}
\item 
$n \in \bN^+$ is the number of states,
\item
$\bfalpha,\bfeta \in \cF^n$ are initial and final vectors, 
respectively,
\item
$\Gamma$ is a finite stack alphabet, and
\item
$\mathrm{M}$ are matrices in $\cF^{n \times n}$ defined as follows.
\begin{itemize}
\item[] 
For every $\gamma \in \Gamma$ and every $c \in \Sigma_{{call}}$, the matrix $\bfM_{{call}}^{(c,\gamma)} \in \cF^{n \times n}$
is given by 
\begin{align*}
{\bfM_{{call}}^{(c, \gamma)}}({i,j}) = &\text{ 
the weight of a $c$-labeled transition from} \\& \text{ state $i$ to state $j$ that pushes $\gamma$ onto the stack.}
\end{align*}
\item[]
The matrices $\bfM_{{ret}}^{(r,\gamma)} \in \cF^{n \times n}$ are given similarly for every $r \in \Sigma_{{ret}}$, and the matrices 
$\bfM_{{int}}^{(s)} \in \cF^{n \times n}$ are given similarly for every $s \in \Sigma_{{int}}$.
\end{itemize}
\end{itemize} 
\end{definition}

\begin{definition}[Behavior of weighted VPA]
\label{def:behavior}
Let $A = (n, \bfalpha, \bfeta, \Gamma, \mathrm{M})$ be an $\cF$-weighted VPA on $\Sigma$. 
For a well-nested word $u \in \WNW(\Sigma)$, the automaton $A$ inductively computes a matrix $\bfM^{(A)}_u \in \cF^{n \times n}$ for $u$ in the following way.
\begin{itemize}
\item Base cases:
\begin{flalign*}
\bfM^{(A)}_\varepsilon = \bfI, \ \text{ and } \ \bfM^{(A)}_s = \bfM_{{int}}^{(s)} \ \text{ for $s \in \Sigma_{{int}}$}.
\end{flalign*}
\item Closure:
\begin{center}
\begin{tabular}{ll}
$\bfM^{(A)}_{uv} \ = \bfM^{(A)}_{u} \cdot \bfM^{(A)}_{v}$ & for $u,v \in \WNW(\Sigma)$, and \\
$\bfM^{(A)}_{c u r} = \sum_{\gamma \in \Gamma}{\bfM_{{call}}^{(c, \gamma)} \cdot \bfM^{(A)}_{u} \cdot \bfM_{{ret}}^{(r, \gamma)}}$ &
for $c \in \Sigma_{{call}}$ and $r \in \Sigma_{{ret}}$.
\end{tabular}
\end{center}

\end{itemize}

The \emph{behavior} of $A$ is the function $f_A: \WNW(\Sigma) \rightarrow \cF$ where 
$$
f_A(u) = \bfalpha^T \cdot \bfM^{(A)}_{u} \cdot \bfeta
$$
\end{definition} 
A function $f: \WNW(\Sigma) \rightarrow \cF$ is \emph{recognizable by weighted VPA} if it is the behavior of some weighted VPA $A$.

%% file: learning.tex
\section{Applications in Computational Learning Theory}
\label{sec:learn}

A passive learning algorithm for classical automata is an algorithm which is given a set of strings accepted by the target automaton (positive examples) and a set of 
strings rejected by the target automaton (negative examples), and is required to output an automaton which is consistent with the set of examples.
It is well known that in a variety of passive learning models, such as Valiant's PAC model, \cite{ar:Valiant84}, and the mistake bound models of Littlestone and  Haussler et al, \cite{ar:littlestone88,ar:Haussler88}, it is intractable to learn or even approximate classical automata, \cite{ar:Gold78,ar:Angluin78,ar:Pitt93}. However, the problem becomes tractable when the learner is allowed to make membership and equivalence queries, as in the active model of learning introduced by Angluin, \cite{ar:Angluin78,ar:Angluin87}. 
This approach was extended to weighted automata over fields, \cite{ar:bergadano96}.

The problem of learning \emph{weighted} automata is of finding a weighted automaton
which closely estimates some target function, by considering examples consisting of pairs of strings with their value. 
The development of efficient learning techniques for weighted automata was immensely motivated 
by the abundance of their applications, with
many of the techniques exploiting the relationship between weighted automata and their Hankel matrices, \cite{ar:Beimel2000,pr:HabrardOncina06,pr:BishtBshoutyMazzawi06}.

\subsection{Learning Weighted Visibly Pushdown Automata}
The proof of our \autoref{th:nested} suggests a template 
of learning algorithms for weighted visibly pushdown automata, with the difficult part being the construction of the matrices that correspond to call and return symbols. 
The proof of \autoref{lem:matrices} spells out the construction of these matrices,
given an algorithm for finding SVD expansions (see \autoref{subs:r2a}) and a spanning set of the nested Hankel matrix.
To the best of our knowledge, learning algorithms for weighted visibly pushdown automata have not been proposed so far.

In recent
years, the spectral method of Hsu et al \cite{ar:hsu12} for learning hidden Markov models, which relies on the SVD of a Hankel matrix, has driven much follow-up research, see the
survey \cite{ar:Balle15}.
Balle and Mohri combined spectral methods with constrained matrix completion algorithms to learn arbitrary weighted automata, \cite{ar:Balle12}. 
We believe the possibility of developing spectral learning algorithms for WVPA is worth exploring in more detail.

Lastly, we should note that one could employ existing algorithms to produce a weighted automaton from a nested Hankel matrix, if it is viewed as a partial Hankel matrix for a word function.
However, any automaton which is consistent with the matrix will have as many states as the rank of the nested Hankel matrix, \cite{ar:CarlylePaz1971,ar:fliess74}. This may be less than satisfying when considering how, in contrast, \autoref{th:nested} assures the existence of a weighted visibly pushdown automaton with $n$ states, given a nested Hankel matrix of rank $\leq n^2$.
 This discrepancy fundamentally depends on the \hyperref[th:SVD]{SVD Theorem}.

%% file: conc.tex
\section{Extension to Semirings}
\label{sec:conc}
The proof of \autoref{th:nested} relies on the 
\hyperref[th:SVD]{SVD theorem}, which, in particular, 
assumes the existence of an inverse with respect to addition.
Furthermore, notions of orthogonality, rank, and norms do not readily transfer to the semiring setting.
Thus it is not clear what an analogue to the SVD theorem would be in the context of semirings, nor whether it could exist.
Therefore the proof of \autoref{th:nested} cannot be used to characterize nested word functions recognized by WVPA over semirings. 

However, in the special case of the tropical semirings, De Schutter and De Moor proposed an extended max algebra corresponding to $\bR$, called the \emph{symmetrized max algebra}, 
and proved an analogue SVD theorem for it, \cite{ar:DeSchutter97}. See also \cite{ar:DeSchutter02} for an extended presentation. These results suggest
a similar Hankel matrix based characterization for WVPA-recognizable nested word functions may be possible over the tropical semirings.
This would be beneficial in situations where we have a function that has a nested Hankel matrix of infinite rank when interpreted over $\bR$, but has finite rank when it is interpreted over a tropical semiring. It is easy to verify that any function on well-nested words which is maximizing or minimizing with respect to concatenation would fall in this category.

%% file: nested.tex
\section{The Characterization of WVPA-Recognizability}
\label{sec:nested}
In this section we prove both directions of \autoref{th:nested}.

\subsection{Recognizability Implies Finite Rank of Nested Hankel Matrix}
This is the easier direction of \autoref{th:nested}.
First we need a definition. For $\ell,m \in [n]$, define the matrix $\bfA^{(\ell,m)} \in \cF ^{n \times n}$ as having the value $1$ in the entry $(\ell,m)$ and zero in all other entries.
That is,
$$
\bfA^{(\ell,m)}(i,j) = 
\begin{cases}
1, & \text{ if } (i,j) = (\ell,m)\\
0, & \text{ otherwise} 
\end{cases}
$$
Obviously, for any matrix $\bfM \in \cF ^{n \times n}$ with entries $\bfM(i,j) = m_{ij}$ we have \\ $\bfM = \sum_{i,j \in [n]}{m_{ij}\bfA^{(i,j)}}$.
\begin{thm}
\label{th:vpa2finite}
Let $f: \WNW(\Sigma) \rightarrow \cF$ be recognized by a weighted visibly pushdown automaton $A$ with $n$ states. 
Then the nested Hankel matrix $\nHa_f$ has rank $\leq n^2$.
\end{thm}
\begin{proof}
We describe infinite row vectors 
$\mathbf{v}^{(i,j)}$ where $i,j \in [n]$, whose entries are indexed by well-nested words $w \in \WNW(\Sigma)$, 
and show they span the rows of $\nHa_f$.
We define the entry of $\bfv^{(i,j)}$ associated with $w$ to be 
$$
\bfv^{(i,j)}(w) = \bfalpha^T \cdot \bfA^{(i,j)} \bfM^{(A)}_{w} \cdot \bfeta
$$
Note that there are $n^2$ such vectors.
Now let $u \in \WNW(\Sigma)$ be a well-nested word and let $\bfM^{(A)}_u$ be the matrix computed for $u$ by $A$
as described in the . 
Then the row $\bfr_u$ corresponding to $u$ 
in $\nHa_f$ has entries 
$$
\bfr_u(w) = \bfalpha^T \cdot \bfM^{(A)}_u \cdot \bfM^{(A)}_w \cdot \bfeta
$$ 
We show this row is linearly spanned by the vectors $\bfv^{(i,j)}$, $i,j \in [n]$.
Consider the linear combination
$$
\bfv_u = \sum_{1 \leq i,j \leq n}{\bfM^{(A)}_{u}(i,j) \cdot \bfv^{(i,j)}}.
$$
Then
\begin{align*}
\bfv_u(w) &= \sum_{1 \leq i,j \leq n}{\bfM^{(A)}_{u}(i,j) \cdot \bfv^{(i,j)}(w)} = \sum_{1 \leq i,j \leq n}{\bfM^{(A)}_{u}(i,j) 
\cdot \lb(\bfalpha^T \cdot \bfA^{(i,j)} \bfM^{(A)}_w \cdot \bfeta\rb)} \\&
= \bfalpha^T \cdot \bfM^{(A)}_u \cdot \bfM^{(A)}_w \cdot \bfeta  =  \bfr_u(w)
\end{align*}
Therefore the rank of $\nHa_f$ is at most $n^2$.
\end{proof}

\section{Finite Rank of Nested Hankel Matrix Implies Recognizability}
\label{subs:r2a}
Here we prove the second direction of \autoref{th:nested}. 
This will be done by defining a weighted VPA which recognizes the function
$f$
described by a given Hankel matrix $\nHa_f$. It will hold that if the rank of $\nHa_f$ is $\leq n^2$, 
the automaton will have at most $n$ states, and a stack alphabet $\Gamma$ of size at most $n$. 

We first describe the initial and final vectors $\bfalpha,\bfeta$ and the matrices that will be used to construct
the automaton, and prove useful properties for them.

As we cannot decompose the Hankel matrix entries in arbitrary positions, but only in ways that maintain the well-nesting,
we will need to use the following theorem to show the matrices used in the construction of the automaton indeed exist:
\begin{thm}[The SVD Theorem, see \cite{bk:Golub12}]
\label{th:SVD}
Let $\bfN \in \cF^{m \times n}$ be a non-zero matrix, where $\cF = \bR$ or $\cF = \bC$. 
Then there exist vectors $\bfx_1,\ldots,\bfx_m \in \cF^m$ and $\bfy_1,\ldots,\bfy_n \in \cF^n$ such that the
matrices
\begin{align*}
\mathbf{X} = [\bfx_1  \ldots  \bfx_m] \in \cF^{m \times m}, \ \ \ \ 
\mathbf{Y} = [\bfy_1  \ldots  \bfy_n] \in \cF^{n \times n}
\end{align*}
are orthogonal, and
$$
\mathbf{Y}^T \bfN \mathbf{X} = diag(\sigma_1, \ldots, \sigma_p) \in \cF^{m \times n}
$$
where $p = \min \{m,n\}$, $diag(\sigma_1, \ldots, \sigma_p)$
is a diagonal matrix with the values $\sigma_1, \ldots, \sigma_p$ on its diagonal,
and $\sigma_1 \geq \sigma_2 \geq \ldots \geq \sigma_p$.
\end{thm}
As a consequence, if we define $r$ by $\sigma_1 \geq \ldots \geq \sigma_r > \sigma_{r+1} = \ldots = 0$, that is the number of non-zero
entries in $diag(\sigma_1, \ldots, \sigma_p)$, then 
we have the \emph{SVD expansion} of $\bfN$:
$$
\bfN = \sum_{i=1}^{r}{\sigma_i \bfx_i \bfy_i^T}
$$
In particular, if $\bfN$ is of rank $1$, then $\bfN = \bfx \bfy^T$.
 
The SVD is perhaps the most important factorization for real and complex matrices. It is used in matrix approximation techniques,
signal processing, computational statistics, and many more areas. See \cite{ar:Klema80,bk:Poularikas10,bk:Gentle09} and references therein.

\subsubsection{The Components of the Automaton}
Let $\Sigma = \Sigma_{{call}} \cup \Sigma_{{ret}} \cup \Sigma_{{int}}$,
where $\Sigma_{{call}}, \Sigma_{{ret}}$ and $\Sigma_{{int}}$ are disjoint.
Throughout this section, let $f: \WNW(\Sigma) \rightarrow \cF$ be a function on well-nested words over $\Sigma$, 
and let its nested Hankel matrix $\nHa_f$ have finite rank $r(\nHa_f) \leq n^2$. 
Denote by $\cB = \{w_{1,1}, \ldots, w_{n,n}\}$ the well-nested words whose rows linearly span $\nHa_f$. 

\begin{definition}[Initial and final vectors]
\label{def:vectors}
Let the matrix $\bfN \in \cF^{n \times n}$ be defined as
$\bfN({i,j}) = f(w_{1,j})$, and let $\bfx, \bfy \in \cF^n$ be vectors such that $\bfN = \bfx \bfy^T$. 
Let 
\begin{equation}
\bfeta = \bfy, \qquad \bfalpha = \bfx
\label{eq:vecs_def}
\end{equation}
\end{definition}
Note that this definition is sound; as $\bfN$ has rank $1$, Theorem \ref{th:SVD} guarantees there exist such
vectors $\bfx, \bfy$.
\begin{definition}[Internal matrices]
\label{def:internal}
For $w_{i,j} \in \cB$, define $\beta_{i,j} = f(w_{i,j})f(w_{1,j})^{-1}$, and let
\begin{equation}
\bfM_{w_{i,j}} = \beta_{i,j} \cdot \bfA^{(i,j)}
\label{eq:w_ij_def}
\end{equation}
Note that for $w_{1,j}$, we have $\beta_{1,j} = 1$ and $\bfM_{w_{1,j}} = \bfA^{(1,j)}$.

For a letter $a \in \Sigma_{int}$,
let $\bfr_a$ denote the row in $\nHa_f$ corresponding to $a$.
The rows indexed by the elements in $\cB$ span the matrix, so there is a linear combination of them equal to $\bfr_a$:
$$
\bfr_a = \sum_{1 \leq i,j \leq n}{z_a^{i,j} \cdot \bfr_{w_{i,j}}}
$$
Set
\begin{equation}
\bfM_a = \sum_{1 \leq i,j\leq n}{z_a^{i,j} \cdot \bfM_{w_{i,j}}}
\label{eq:int_a_def}
\end{equation}
\end{definition}
Before we define the call and return matrices, we show the above defined vectors and matrices behave as expected:
\begin{lem}
\label{lem:def_first}
Let $\bfalpha,\bfeta \in \cF^n$, $\bfM_{w_{i,j}} \in \cF^{n \times n}$ for $w_{i,j} \in \cB$, and 
$\bfM_a$ for $a \in \Sigma_{{int}}$ be defined as in Definitions \ref{def:vectors} and \ref{def:internal}.
It holds that:
\begin{align}
f(w_{i,j}) &= \bfalpha^T \cdot \bfM_{w_{i,j}} \cdot \bfeta
\label{eq:lm1}
\\
f(a) &= \bfalpha^T \cdot \bfM_a \cdot \bfeta
\label{eq:lm2}
\end{align}
\end{lem}

\begin{proof}
Since the entries of $\bfM_{w_{i,j}}$ are zero except for entry $(i,j)$, we have
\begin{equation*}
\bfalpha^T \cdot \bfM_{w_{i,j}} \cdot \bfeta =  \bfalpha(i) \cdot \beta_{i,j} \cdot \bfeta(j)
\end{equation*} 
Since $\bfalpha(i)\bfeta(j) = f(w_{1,j})$, we have
\begin{align*}
\bfalpha^T \cdot \bfM_{w_{i,j}} \cdot \bfeta &= \beta_{i,j} \cdot f(w_{1,j}) = f(w_{i,j}) \cdot f(w_{1,j})^{-1} \cdot f(w_{1,j}) = f(w_{i,j})
\end{align*}
and Equation \ref{eq:lm1} holds.

By the definition of $\bfM_a$, we have
\begin{align*}
\bfalpha^T \cdot \bfM_a \cdot \bfeta &
= \bfalpha^T \lb(\sum_{1 \leq i,j\leq n}{ z_a^{i,j} \cdot \bfM_{w_{i,j}}}\rb) \bfeta 
=
\sum_{1 \leq i,j\leq n}{z_a^{i,j} \lb(\bfalpha^T \bfM_{w_{i,j}} \bfeta\rb)} \\&
=
\sum_{1 \leq i,j\leq n}{z_a^{i,j} \cdot f(w_{i,j})}
= f(a)
\end{align*}
and Equation \ref{eq:lm2} holds.
\end{proof}

\begin{definition}[Call and return matrices]
\label{def:call_ret_def}
For each pair $c \in \Sigma_{call}$ and $r \in \Sigma_{ret}$, define the $n \times n$ matrix $\bfN_{c,r}$ as
$\bfN_{c,r}(i,j) = f(c w_{i,j} r)/\beta_{i,j}$
and let its SVD be
$$
\bfN_{c,r} = \sum_{k=1}^{n}{\bfp_{c,k} (\bfp_{r,k})^T}
$$
For $\gamma \in \Gamma$, define
$$
\bfM_{call}^{c,\gamma}(\ell,i)
=
\begin{cases}
\bfp_{c,\gamma}(i)/\bfalpha(\gamma) & \ell=\gamma \\
0 & \text{else}
\end{cases}
$$
and 
$$
\bfM_{ret}^{r,\gamma}(j,m)
=
\begin{cases}
\bfp_{r,\gamma}(j)/\bfeta(\gamma)  & m=\gamma \\
0 & \text{else}
\end{cases}
$$

\end{definition}
Now we show these matrices behave as expected:
\begin{lem}
\label{lem:matrices}
Let $\bfM_{call}^{(c,\gamma)}$ and $\bfM_{ret}^{(r,\gamma)}$ for  $c\in \Sigma_{{call}}$, $r \in \Sigma_{{ret}}$ be defined
as in Definition \ref{def:call_ret_def}. Then it holds that:
\begin{equation}
f(c w_{i,j} r) = \bfalpha^T \left(\sum_{\gamma =1}^{n}{ \bfM_{call}^{(c,\gamma)} \cdot \bfM_{w_{i,j}} \cdot \bfM_{ret}^{(r,\gamma)} }\right)\bfeta
\label{eq:lm3}
\end{equation}
\end{lem}

\begin{proof}
By the definition of $\bfN_{c,r}$ we have:
\begin{equation}
\bfN_{c,r}(i,j) = f(c w_{i,j} r)/\beta_{i,j} = \sum_{k=1}^{n}{\bfp_{c,k}(i) \bfp_{r,k}(j)}
\label{eq:2}
\end{equation}
In addition, $\bfM_{w_{ij}}$ is zero in all entries that are not the $(i,j)$ one. Therefore,
$$
(\bfM_{call}^{c,\gamma} \cdot \bfM_{w_{ij}})(\ell,m) = 
\begin{cases}
\bfM_{call}^{c,\gamma}(\gamma,i) \beta_{i,j} = \bfp_{c,\gamma}(i)\beta_{i,j}/ \bfalpha(\gamma) & \ell=\gamma, m=j \\
0 & \text{else}
\end{cases}
$$

Thus multiplying with $\bfM_{ret}^{r,\gamma}$ results in:
$$
(\bfM_{call}^{c,\gamma} \cdot \bfM_{w_{ij}} \cdot \bfM_{ret}^{r,\gamma})(\ell,m)
= 
\begin{cases}
(\bfp_{c,\gamma}(i) \beta_{i,j} \bfp_{r,\gamma}(j))/(\bfalpha(\gamma)\bfeta(\gamma)) & \ell=m=\gamma \\
0 & \text{else}
\end{cases}
$$
Note that the above matrix $\bfM_{call}^{c,\gamma} \cdot \bfM_{w_{ij}} \cdot \bfM_{ret}^{r,\gamma}$ is diagonal.
Therefore, in total:
\begin{align*}
\bfalpha^T \left(\sum_{\gamma=1}^{n}{\bfM_{call}^{c,\gamma} \bfM_{w_{i,j}} \bfM_{ret}^{r,\gamma}}\right) \bfeta &= 
\beta_{i,j} \sum_{\gamma=1}^{n}{\bfp_{c,\gamma}(i)\bfp_{r,\gamma}(j)} 
\\&\underset{Equation \ref{eq:2}}{=}
\beta_{i,j} (f(c w_{i,j} r)/\beta_{i,j}) 
= f(c w_{i,j} r)
\end{align*}

so Equation \ref{eq:lm3} holds.
\end{proof}

\subsubsection{Putting the Automaton Together}
We are now ready to prove the second direction of \autoref{th:nested}:
\begin{thm}
\label{th:finite2vpa}
Let $f: \WNW(\Sigma) \rightarrow \cF$ have a nested Hankel matrix $\nHa_f$ of rank $\leq n^2$. Then $f$ is recognizable by a weighted visibly pushdown automaton $A$ with $n$ states.
\end{thm}
\begin{proof}
Use Definitions \ref{def:vectors}, \ref{def:internal}, and \ref{def:call_ret_def} to build a weighted VPA $A$ with $n$ states, and set $\bfM^{(A)}_\varepsilon = \bfI$. 
From Lemmas \ref{lem:def_first} and \ref{lem:matrices}
it only remains to show that for $u,t \in \WNW(\Sigma)$,
$$
\bfM^{(A)}_{ut} = \bfM^{(A)}_u \cdot \bfM^{(A)}_t
$$ 

Note that we defined the matrices $\bfM^{(A)}_{w_{i,j}}$ such
that $\bfr_{w_{i,j}} = \bfv^{(i,j)}$ up to a constant factor. 
We show that if 
\begin{align*}
\bfr_{u} = \sum_{1 \leq i,j \leq n}{\bfM^{(A)}_u(i,j) \cdot \bfv^{(i,j)}}
\quad \text{and} \quad
\bfr_t = \sum_{1 \leq i,j \leq n}{\bfM^{(A)}_t(i,j) \cdot \bfv^{(i,j)}},
\end{align*}
then 
$$
\bfr_{ut} = \sum_{1 \leq i,j \leq n}{(\bfM^{(A)}_u \cdot \bfM^{(A)}_t)(i,j) \cdot \bfv^{(i,j)}}
$$
Or, equivalently, that for every well-nested word $w \in \WNW(\Sigma)$,
$$
\bfr_{ut}(w) = \bfalpha^T \cdot \bfM^{(A)}_u \cdot \bfM^{(A)}_t \cdot \bfM^{(A)}_w \cdot \bfeta
$$
Consider the linear combination:
\begin{align*}
\bfv_{ut} &= \sum_{1 \leq i,j \leq n}{ (\bfM^{(A)}_u \cdot \bfM^{(A)}_t)(i,j) \cdot \bfv^{(i,j)}}
 = \sum_{1 \leq i,k,j \leq n}{ \bfM^{(A)}_u(i,k) \cdot \bfM^{(A)}_t(k,j) \cdot \bfv^{(i,j)}}
\end{align*}
Then, for $w \in \WNW(\Sigma)$ we have
\begin{align*}
\bfv_{ut}(w) &= \sum_{1 \leq i,k,j \leq n}{ \bfM^{(A)}_u(i,k) \cdot \bfM^{(A)}_t(k,j) \cdot \bfv^{(i,j)}(w)} \\&= 
\sum_{1 \leq i,k,j \leq n}{\bfM^{(A)}_u(i,k) \cdot \bfM^{(A)}_t(k,j) \cdot \lb(\bfalpha^T \cdot \bfA^{(i,j)} \bfM^{(A)}_w \cdot \bfeta\rb)} &
\end{align*}
Note that the row $i$ of $\bfA^{(i,j)} \bfM^{(A)}_w$ is row $j$ of $\bfM^{(A)}_w$ and all other rows are zero.
Then
\begin{align*}
\bfv_{ut}(w) &= \sum_{1 \leq i,k,j \leq n}{\bfM^{(A)}_u(i,k) \cdot \bfM^{(A)}_t(k,j) \cdot 
\lb(\sum_{l=1}^{n}{\bfalpha(i) \cdot \bfM^{(A)}_w(j,l) \cdot \bfeta(l)}\rb)}
\\&
= \sum_{1 \leq i,k,j,l \leq n}{\bfalpha(i) \cdot \bfM^{(A)}_u(i,k) \cdot \bfM^{(A)}_t(k,j) \cdot \bfM^{(A)}_w(j,l) \cdot \bfeta(l)}
\\&
= \bfalpha^T \cdot \bfM^{(A)}_u \cdot \bfM^{(A)}_t \cdot \bfM^{(A)}_w \cdot \bfeta =  \bfr_{ut}(w)
\end{align*}
\end{proof}

From Theorem \ref{th:finite2vpa} and Theorem \ref{th:vpa2finite} we have our main result, Theorem \ref{th:nested}.

%% file: lata2016.bbl
\begin{thebibliography}{10}

\bibitem{ar:Allauzen04}
C.~Allauzen, M.~Mohri, and M.~Riley.
\newblock Statistical modeling for unit selection in speech synthesis.
\newblock In {\em Proceedings of the 42nd Annual Meeting on Association for
  Computational Linguistics}, page~55. Association for Computational
  Linguistics, 2004.

\bibitem{alur2007first}
R.~Alur, M.~Arenas, P.~Barcel{\'o}, K.~Etessami, N.~Immerman, and L.~Libkin.
\newblock First-order and temporal logics for nested words.
\newblock In {\em Logic in Computer Science, 2007. LICS 2007. 22nd Annual IEEE
  Symposium on}, pages 151--160. IEEE, 2007.

\bibitem{ar:AlurMadhusudan06}
R.~Alur and P.~Madhusudan.
\newblock Adding nesting structure to words.
\newblock In {\em Developments in Language Theory}, pages 1--13. Springer,
  2006.

\bibitem{ar:Angluin78}
D.~Angluin.
\newblock On the complexity of minimum inference of regular sets.
\newblock {\em Information and Control}, 39(3):337--350, 1978.

\bibitem{ar:Angluin87}
D.~Angluin.
\newblock Learning regular sets from queries and counterexamples.
\newblock {\em Information and computation}, 75(2):87--106, 1987.

\bibitem{bk:Arnold1994}
A.~Arnold and J.~Plaice.
\newblock {\em Finite transition systems: semantics of communicating systems}.
\newblock Prentice Hall International (UK) Ltd., 1994.

\bibitem{ar:Balle12}
B.~Balle and M.~Mohri.
\newblock Spectral learning of general weighted automata via constrained matrix
  completion.
\newblock In {\em Advances in neural information processing systems}, pages
  2168--2176, 2012.

\bibitem{ar:Balle15}
B.~Balle and M.~Mohri.
\newblock Learning weighted automata.
\newblock In {\em Algebraic Informatics}, pages 1--21. Springer, 2015.

\bibitem{ar:Beimel2000}
A.~Beimel, F.~Bergadano, N.~Bshouty, E.~Kushilevitz, and S.~Varricchio.
\newblock Learning functions represented as multiplicity automata.
\newblock {\em Journal of the ACM (JACM)}, 47(3):506--530, 2000.

\bibitem{ar:bergadano96}
F.~Bergadano and S.~Varricchio.
\newblock Learning behaviors of automata from multiplicity and equivalence
  queries.
\newblock {\em SIAM Journal on Computing}, 25(6):1268--1280, 1996.

\bibitem{pr:BishtBshoutyMazzawi06}
L.~Bisht, N.~Bshouty, and H.~Mazzawi.
\newblock {\em On optimal learning algorithms for multiplicity automata}.
\newblock Springer, 2006.

\bibitem{ar:CarlylePaz1971}
J.~Carlyle and A.~Paz.
\newblock Realizations by stochastic finite automata.
\newblock {\em J. Comp. Syst. Sc.}, 5:26--40, 1971.

\bibitem{pr:ChatterjeeDH09}
K.~Chatterjee, L.~Doyen, and T.~Henzinger.
\newblock Probabilistic weighted automata.
\newblock In {\em CONCUR 2009-Concurrency Theory}, pages 244--258. Springer,
  2009.

\bibitem{pr:ChatterjeeHJS10}
K.~Chatterjee, T.~Henzinger, B.~Jobstmann, and R.~Singh.
\newblock Measuring and synthesizing systems in probabilistic environments.
\newblock In {\em Computer Aided Verification}, pages 380--395. Springer, 2010.

\bibitem{ar:Cobham1978}
A.~Cobham.
\newblock Representation of a word function as the sum of two functions.
\newblock {\em Mathematical Systems Theory}, 11:373--377, 1978.

\bibitem{bk:CourcelleEngelfriet2011}
B.~Courcelle and J.~Engelfriet.
\newblock {\em Graph structure and monadic second-order logic: a
  language-theoretic approach}, volume 138.
\newblock Cambridge University Press, 2012.

\bibitem{pr:CourcelleMakowskyRoticsWG98}
B.~Courcelle, J.~Makowsky, and U.~Rotics.
\newblock Linear time solvable optimization problems on graph of bounded clique
  width, extended abstract.
\newblock In J.~Hromkovic and O.~Sykora, editors, {\em Graph Theoretic Concepts
  in Computer Science, 24th International Workshop, WG'98}, volume 1517 of {\em
  Lecture Notes in Computer Science}, pages 1--16. Springer Verlag, 1998.

\bibitem{ar:CourcelleMakowskyRoticsDAM}
B.~Courcelle, J.~Makowsky, and U.~Rotics.
\newblock On the fixed parameter complexity of graph enumeration problems
  definable in monadic second order logic.
\newblock {\em Discrete Applied Mathematics}, 108(1-2):23--52, 2001.

\bibitem{ar:CulikKari93}
K.~Culik~II and J.~Kari.
\newblock Image compression using weighted finite automata.
\newblock In {\em Mathematical Foundations of Computer Science 1993}, pages
  392--402. Springer, 1993.

\bibitem{d2014symbolic}
L.~D'Antoni and R.~Alur.
\newblock Symbolic visibly pushdown automata.
\newblock In {\em Computer Aided Verification}, pages 209--225. Springer, 2014.

\bibitem{ar:DeSchutter97}
B.~De~Schutter and B.~De~Moor.
\newblock The singular-value decomposition in the extended max algebra.
\newblock {\em Linear Algebra and Its Applications}, 250:143--176, 1997.

\bibitem{ar:DeSchutter02}
B.~De~Schutter and B.~De~Moor.
\newblock The {QR} decomposition and the singular value decomposition in the
  symmetrized max-plus algebra revisited.
\newblock {\em SIAM review}, 44(3):417--454, 2002.

\bibitem{bk:DowneyFellows99}
R.~Downey and M.~Fellows.
\newblock {\em Parametrized Complexity}.
\newblock Springer, 1999.

\bibitem{driscoll2011checking}
E.~Driscoll, A.~Burton, and T.~Reps.
\newblock Checking compatibility of a producer and a consumer.
\newblock Citeseer, 2011.

\bibitem{driscoll2012opennwa}
E.~Driscoll, A.~Thakur, and T.~Reps.
\newblock Opennwa: A nested-word automaton library.
\newblock In {\em Computer Aided Verification}, pages 665--671. Springer, 2012.

\bibitem{ar:DrosteGastin05}
M.~Droste and P.~Gastin.
\newblock Weighted automata and weighted logics.
\newblock In {\em ICALP 2005}, pages 513--525, 2005.

\bibitem{bk:DrosteKuichVogler2009}
M.~Droste, W.~Kuich, and H.~Vogler.
\newblock {\em Handbook of weighted automata}.
\newblock Springer Science \& Business Media, 2009.

\bibitem{ar:Fernando97}
C.~Fernando, N.~Pereira, and M.~Riley.
\newblock Speech recognition by composition of weighted finite automata.
\newblock {\em Finite-State Language Processing. MIT Press, Cambridge,
  Massachusetts}, 1997.

\bibitem{ar:fliess74}
M.~Fliess.
\newblock Matrices de hankel.
\newblock {\em J. Math. Pures Appl}, 53(9):197--222, 1974.

\bibitem{ar:FreedmanLovaszSchrijver07}
M.~Freedman, L.~Lov\'asz, and A.~Schrijver.
\newblock Reflection positivity, rank connectivity, and homomorphism of graphs.
\newblock {\em Journal of the American Mathematical Society}, 20(1):37--51,
  2007.

\bibitem{gauwin2011streamable}
O.~Gauwin and J.~Niehren.
\newblock Streamable fragments of forward xpath.
\newblock In {\em Implementation and Application of Automata}, pages 3--15.
  Springer, 2011.

\bibitem{bk:Gentle09}
J.~Gentle.
\newblock {\em Computational statistics}, volume 308.
\newblock Springer, 2009.

\bibitem{ar:GodlinKotekMakowsky08}
B.~Godlin, T.~Kotek, and J.~Makowsky.
\newblock Evaluation of graph polynomials.
\newblock In {\em 34th International Workshop on Graph-Theoretic Concepts in
  Computer Science, WG08}, volume 5344 of {\em Lecture Notes in Computer
  Science}, pages 183--194, 2008.

\bibitem{ar:Gold78}
E.~Gold.
\newblock Complexity of automaton identification from given data.
\newblock {\em Information and control}, 37(3):302--320, 1978.

\bibitem{bk:Golub12}
G.~Golub and C.~Van~Loan.
\newblock {\em Matrix computations}, volume~3.
\newblock JHU Press, 2012.

\bibitem{pr:HabrardOncina06}
A.~Habrard and J.~Oncina.
\newblock Learning multiplicity tree automata.
\newblock In {\em Grammatical Inference: Algorithms and Applications}, pages
  268--280. Springer, 2006.

\bibitem{harris2012secure}
W.~R. Harris, S.~Jha, and T.~Reps.
\newblock Secure programming via visibly pushdown safety games.
\newblock In {\em Computer Aided Verification}, pages 581--598. Springer, 2012.

\bibitem{ar:Haussler88}
D.~Haussler, N.~Littlestone, and M.~Warmuth.
\newblock Predicting $\{$0, 1$\}$-functions on randomly drawn points.
\newblock In {\em Foundations of Computer Science, 1988., 29th Annual Symposium
  on}, pages 100--109. IEEE, 1988.

\bibitem{ar:Heller1967}
A.~Heller.
\newblock Probabilistic automata and stochastic transformations.
\newblock {\em Theory of Computing Systems}, 1(3):197--208, 1967.

\bibitem{ar:hsu12}
D.~Hsu, S.~Kakade, and T.~Zhang.
\newblock A spectral algorithm for learning hidden markov models.
\newblock {\em Journal of Computer and System Sciences}, 78(5):1460--1480,
  2012.

\bibitem{ar:Kiefer12}
S.~Kiefer, A.~S. Murawski, J.~Ouaknine, B.~Wachter, and J.~Worrell.
\newblock On the complexity of equivalence and minimisation for {Q}-weighted
  automata.
\newblock {\em Logical Methods in Computer Science (LMCS)}, 9(1:8):1--22, 2013.

\bibitem{ar:Klema80}
V.~Klema and A.~Laub.
\newblock The singular value decomposition: Its computation and some
  applications.
\newblock {\em Automatic Control, IEEE Transactions on}, 25(2):164--176, 1980.

\bibitem{msc:Labai}
N.~Labai.
\newblock Definability and hankel matrices.
\newblock Master's thesis, Technion - Israel Institute of Technology, Faculty
  of Computer Science, 2015.

\bibitem{pr:LabaiMakowsky2013}
N.~Labai and J.~Makowsky.
\newblock Weighted automata and monadic second order logic.
\newblock {\em EPTCS Proceedings of GandALF}, 119:122--135, 2013.

\bibitem{ar:LabaiMakowsky2014}
N.~Labai and J.~Makowsky.
\newblock Tropical graph parameters.
\newblock {\em DMTCS Proceedings of FPSAC}, (01):357--368, 2014.

\bibitem{ar:LabaiMakowskyJCSS}
N.~Labai and J.~Makowsky.
\newblock Meta-theorems using hankel matrices.
\newblock 2015.

\bibitem{ar:littlestone88}
N.~Littlestone.
\newblock Learning quickly when irrelevant attributes abound: A new
  linear-threshold algorithm.
\newblock {\em Machine learning}, 2(4):285--318, 1988.

\bibitem{ar:Lovasz07}
L.~Lov{\'a}sz.
\newblock Connection matrices.
\newblock {\em OXFORD LECTURE series IN MATHEMATICS AND ITS APPLICATIONS},
  34:179, 2007.

\bibitem{bk:Lovasz-hom}
L.~Lov\'asz.
\newblock {\em Large Networks and Graph Limits}, volume~60 of {\em Colloquium
  Publications}.
\newblock AMS, 2012.

\bibitem{ar:MakowskyTARSKI}
J.~Makowsky.
\newblock Algorithmic uses of the {F}eferman-{V}aught theorem.
\newblock {\em Annals of Pure and Applied Logic}, 126.1-3:159--213, 2004.

\bibitem{ar:Mathissen08}
C.~Mathissen.
\newblock Weighted logics for nested words and algebraic formal power series.
\newblock In {\em Automata, Languages and Programming}, pages 221--232.
  Springer, 2008.

\bibitem{bk:McMillan1993}
K.~McMillan.
\newblock {\em Symbolic model checking}.
\newblock Springer, 1993.

\bibitem{ar:Mohri97}
M.~Mohri.
\newblock Finite-state transducers in language and speech processing.
\newblock {\em Computational linguistics}, 23(2):269--311, 1997.

\bibitem{mozafari2012high}
B.~Mozafari, K.~Zeng, and C.~Zaniolo.
\newblock High-performance complex event processing over xml streams.
\newblock In {\em Proceedings of the 2012 ACM SIGMOD International Conference
  on Management of Data}, pages 253--264. ACM, 2012.

\bibitem{murawski2005third}
A.~S. Murawski and I.~Walukiewicz.
\newblock Third-order idealized algol with iteration is decidable.
\newblock In {\em Foundations of Software Science and Computational
  Structures}, pages 202--218. Springer, 2005.

\bibitem{ar:Pitt93}
L.~Pitt and M.~Warmuth.
\newblock The minimum consistent dfa problem cannot be approximated within any
  polynomial.
\newblock {\em Journal of the ACM (JACM)}, 40(1):95--142, 1993.

\bibitem{bk:Poularikas10}
A.~Poularikas.
\newblock {\em Transforms and applications handbook}.
\newblock CRC press, 2010.

\bibitem{ar:Valiant84}
L.~Valiant.
\newblock A theory of the learnable.
\newblock {\em Communications of the ACM}, 27(11):1134--1142, 1984.

\end{thebibliography}
